\documentclass[journal,twocolumn]{IEEEtran}
\usepackage[dvipdfmx]{graphicx}

\usepackage{amsmath}
\usepackage{amssymb}
\usepackage{color}
\DeclareMathAlphabet{\bm}{OML}{cmm}{b}{it}
\newtheorem{theorem}{Theorem}
\newtheorem{lemma}[theorem]{Lemma}
\newtheorem{definition}[theorem]{Definition}
\newtheorem{corollary}[theorem]{Corollary}
\newtheorem{remark}[theorem]{Remark}
\newtheorem{proposition}[theorem]{Proposition}
\newcommand{\qed}{\hfill \IEEEQED}

\newcommand{\markov}{\leftrightarrow}
\newcommand{\bol}[1]{\mathbf{#1}}

\allowdisplaybreaks[2]

\begin{document}

\title{Broadcast Channels with Confidential Messages by Randomness Constrained Stochastic 
Encoder}

\author{Shun~Watanabe~\IEEEmembership{Member,~IEEE}       
\thanks{The first author is with the Department
of Information Science and Intelligent Systems, 
University of Tokushima,
2-1, Minami-josanjima, Tokushima,
770-8506, Japan, 
e-mail:shun-wata@is.tokushima-u.ac.jp.}
and Yasutada~Oohama~\IEEEmembership{Member,~IEEE}       
\thanks{The second author is with the Department of Communication Engineering and 
Informatics,
University of Electro-Communications, Tokyo, 182-8585, Japan, e-mail:oohama@uec.ac.jp.}

\thanks{Manuscript received ; revised }}

\markboth{Journal of \LaTeX\ Class Files,~Vol.~6, No.~1, January~2007}%
{Shell \MakeLowercase{\textit{et al.}}: Bare Demo of IEEEtran.cls for Journals}

\maketitle
\begin{abstract}
%\boldmath
In coding schemes for the wire-tap channel or the broadcast channels with 
confidential messages, it is well known that the sender needs to use a stochastic 
encoding to avoid the information about the transmitted confidential message to be leaked 
to an
eavesdropper. In this paper, it is investigated that the trade-off between the rate of the 
random number to
realize the stochastic encoding and the rates of the common, private, and confidential
messages. For the direct theorem, the superposition coding
scheme for the wire-tap channel recently proposed by Chia and El Gamal
is employed, and its strong security is proved.
The matching converse theorem is also established. 
Our result clarifies that a combination of the ordinary stochastic encoding
and the channel prefixing by the channel simulation
is suboptimal.
\end{abstract}


\IEEEpeerreviewmaketitle

\section{Introduction}

The wire-tap channel is one sender and two receivers broadcast channel model
in which the sender, usually referred to as Alice, wants to transmit a confidential message to 
the legitimate 
receiver, usually referred to as Bob, in such a way that the other receiver, usually
referred to as eavesdropper Eve, cannot get any information about the
transmitted message. The wire-tap channel model 
was first introduced by Wyner in his seminal paper \cite{wyner:75}.
Later, Csisz\'ar and K\"orner
investigated the model called broadcast channels with confidential messages (BCC)
in which Alice also sends a common message that is supposed to be decoded by both Bob 
and Eve.
These models were further investigated by many researchers
from theoretical point of view (e.g..~see \cite{liang-book}), and recently it has attracted 
considerable attention from practical point of view as 
a physical layer security.

In coding schemes for the wire-tap channel or the BCC, it is well known that the sender 
needs to use a
stochastic encoder to avoid the information about the transmitted confidential message 
to be leaked to Eve. The stochastic encoding is usually realized by preparing
a dummy random number in addition to the intended messages and by encoding
them to a transmitted signal by a deterministic encoder. 
Furthermore, when the channel to Bob is not more capable than 
the channel to Eve, it is known that the sender needs to use 
the channel prefixing to achieve the capacity region (or the secrecy capacity)
because the capacity formulas involve such a channel from
an auxiliary random variable to the random variable describing the input signal of the 
channel
\cite{csiszar:78}. In literatures, it is assumed that there exists a channel 
realizing the channel prefixing. But in practice the prefixing channel must be
simulated from a random number by using a method such as 
the channel simulation \cite{steinberg:94}, which usually involves
certain amount of simulation error 
depending on the amount of the random number.
So far, there was no paper investigating how much random number
is needed to achieve the capacity region.
Since the random number is precious resource in practice,
though it has been paid no attention in literatures,
it is extremely important to investigate the amount of random number
needed to achieve the capacity region.
For this purpose, we formulate the problem of the BCC by
randomness constrained stochastic encoder, and completely
characterize the capacity region of this new problem.

The present problem to consider the randomness constrained stochastic encoder
is motivated by the authors' previous
results in \cite{oohama:10}. In that paper, the capacity
region of the relay channel with confidential messages 
for the completely deterministic encoder was investigated,
and the capacity region was characterized
for the BCC as a corollary. In this paper, we are interested in the 
case such that the randomness is constrained but not zero.
The result in \cite{oohama:10} can be regarded as an extreme case
of the present problem. On the other hand, the conventional BCC
problem can be regarded as the other extreme case, in which
the amount of randomness that can be used at the encoder is unbounded.

Typically in the BCC, Alice sends the common message that is 
supposed to be decoded by both Bob and Eve, and the confidential message
that is supposed to be decoded only by Bob. The level of secrecy of the 
confidential message is usually evaluated by the equivocation rate.
In this paper, we consider slightly different problem formulation, which has
been appeared in the literature \cite{csiszar-korner:11,oohama:10}.
In our problem setting, Alice sends three kinds of messages,
the common message, the private message, and the confidential messages.
The common message is supposed to be decoded by both Bob and Eve.
The private message is supposed to be decoded by Bob, and we do not 
care whether Eve can decode the private message or not.
The confidential message is supposed to be decoded by Bob, and it
must be kept completely secret from Eve.
Furthermore, for stochastic encoding, Alice is allowed to use
limited amount of dummy randomness. Thus, we are interested in
the trade-off between quadruple of rates, the rate of dummy randomness,
the rates of common, private, and confidential messages.
The coding system of our formulation is depicted in Fig.~\ref{Fig:system}.

%%%%%%%% Fig %%%%%%%%%%%%%%%
\begin{figure}[th]
\centering
\includegraphics[width=\linewidth]{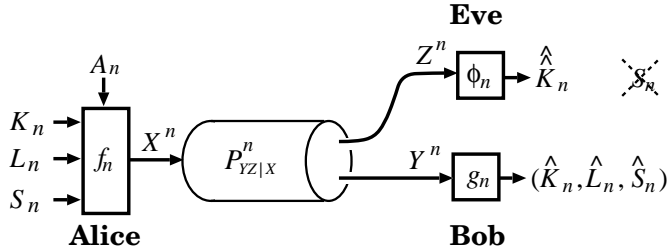}
\caption{The coding system investigated in this paper. Alice sends common
message $K_n$, private message $L_n$, and confidential message $S_n$
by using a deterministic function $f_n$ and a limited amount of dummy randomness $A_n$.
The common message is supposed to be decoded by both Bob and Eve.
The private message is supposed to be decoded by Bob, and we do not 
care whether Eve can decode the private message or not.
The confidential message is supposed to be decoded by Bob, and it
must be kept completely secret from Eve.
}
\label{Fig:system}
\end{figure}

The  reason we do not use the equivocation rate formulation is as follows.
In the conventional equivocation rate formulation, if the rate of dummy randomness
is not sufficient, a part of the confidential message is sacrificed to make the other
part completely secret and the rate of the completely secret part corresponds
to the equivocation rate. We  think that the rates of sacrificed part and completely secret part
become clearer by employing our formulation.

As we have mentioned above, the typical coding scheme for the wire-tap channel
or the BCC consists of the stochastic encoding 
and the channel prefixing. In \cite{chia:12}, Chia and
El Gamal proposed an alternative coding scheme
that utilizes the so-called superposition coding \cite{cover} instead
of the channel prefixing. In the direct part of our main result, we employ this superposition
scheme instead of the channel prefixing.
We also clarifies that a straightforward combination of the ordinary 
stochastic coding and the channel prefixing by the channel simulation method
is suboptimal.

Although Chia and El Gamal showed that the superposition coding scheme
can realize the so-called weak security criterion \cite{maurer:94b, csiszar:96},
it was not clear whether the superposition coding scheme can 
realize the so-called strong security criterion.
One of technical contributions of this paper is to show that
Chia and El Gamal's superposition coding scheme can
realize the strong security criterion.
 This is done by using the technique
proposed in \cite{hayashi:10}, and by
considering the channel resolvability
problem \cite{han:93} with the superposition coding.
Note that the relationship between the wire-tap channel coding and 
the channel resolvability was first pointed out 
by Csisz\'ar \cite{csiszar:96}, and is well recognized 
recently \cite{cai:04,devetak:05,hayashi:06b,bloch:11}.
The channel resolvability with the superposition coding was
first investigated by the second author in \cite{oohama:03}.
In that paper, the channel resolvability problem with the superposition
coding for the degraded
broadcast channel was considered to show the converse theorem
of the identification via degraded broadcast channels.
In this paper, the channel resolvability problem with 
the superposition coding for a single channel is considered.
Using the superposition coding for a single channel seems
nonsense at first glance, it does have a meaning when applied
to the wire-tap channel or the BCC.

After the submission of the first manuscript of this paper,
we noticed some related works investigating the importance of
random number in the BCC or the wire-tap channel.
In \cite{hayashi:12}, Hayashi and Matsumoto considered 
the secure multiplex coding \cite{kobayashi:05} in which
the messages are not necessarily uniform nor independent and
the entropy rate of the messages might be constrained. 
Although the secure multiplex coding can be regarded as a generalization
of the BCC, the encoder can use unlimited amount of dummy randomness
in addition to the messages in their problem formulation.
Thus, our results cannot be derived from their results.
In \cite{bloch:12},
Bloch and Kliewer  
considered the wire-tap channel in which the dummy 
randomness is constrained and not necessarily uniform.  
However, they only considered the case such that
the channel to Bob is more capable than that to Eve.
In such a case, the channel prefixing is not needed,
and their result corresponds to 
Corollary \ref{corollary-more-capable} in this paper
when the dummy randomness is uniform.

The rest of the paper is organized as follows.
In Section \ref{section:problem-main}, the problem
formulation is explained and main results are
presented. In Section \ref{section:resolvability},
the channel resolvability problem with the 
superposition coding is discussed. In
Section \ref{proof-of-main-result}, the proof of the main theorem
is presented.
In Section \ref{section:conclusion}, the paper
is concluded with discussions.
Some technical arguments are presented in Appendices.

%%%%%%% Problem and Main %%%%%%

\section{Problem Formulation and Main Results}
\label{section:problem-main}

\subsection{Problem Formulation}

Let $P_{Y|X}$ and $P_{Z|X}$ be two channels with common
input alphabet ${\cal X}$ and output alphabets ${\cal Y}$ and
${\cal Z}$ respectively. Throughout the paper, the alphabets are assumed to
be finite though we do not use finiteness of the alphabet except
cardinality bonds on auxiliary random variables. 
We also assume that the base of logarithm is $e$ throughout the paper.

%In the broadcast channels with confidential messages \cite{csiszar:78},
%the sender, Alice, sends a common message and a confidential message
%to the legitimate receiver, Bob and the eavesdropper, Eve. 
%The common message is supposed to be decoded by both
%Bob and Eve, and the confidential message is supposed to be
%decoded only by Bob, and it must be kept as secret as possible from Eve.
%Although the level of secrecy is usually evaluated by the equivocation
%of the confidential message for Eve, we consider a slightly different setting
%concerning secrecy, which has been appeared in \cite{csiszar-korner:11,oohama:10}.
Let ${\cal K}_n$ be the set of the common message, ${\cal L}_n$ be the set of
the private message, and ${\cal S}_n$ be the set of the confidential message.
The common message is supposed to be decoded by both Bob and Eve.
The private message is supposed to be decoded by Bob, and we do not care
whether Eve can decode the private message or not. The confidential message
is supposed to be decoded by Bob, and it must be kept completely secret from Eve.

Typically, Alice use a stochastic encoder to make the confidential message secret
from Eve, and it is practically realized by using a uniform dummy randomness on
the alphabet ${\cal A}_n$. When the size $|{\cal A}_n|$ of dummy randomness is
infinite, any stochastic encoder from ${\cal K}_n \times {\cal L}_n \times {\cal S}_n$ to 
${\cal X}^n$
can be simulated by a deterministic encoder 
$f_n:{\cal K}_n \times {\cal L}_n \times {\cal S}_n \times {\cal A}_n \to {\cal X}^n$.
But we are interested in the case with bounded size  $|{\cal A}_n|$ in this paper.

Bob's decoder is defined by function $g_n: {\cal Y}^n \to {\cal K}_n \times {\cal L}_n \times 
{\cal S}_n$
and the error probability is defined as 
\begin{eqnarray}
\lefteqn{ P_{err}(f_n,g_n) } \nonumber \\ 
	&=& \sum_{k_n \in {\cal K}_n} \sum_{\ell_n \in {\cal L}_n} \sum_{s_n \in {\cal S}_n}
	\sum_{a_n \in {\cal A}_n} \frac{1}{|{\cal K}_n||{\cal L}_n||{\cal S}_n||{\cal A}_n|} 
\nonumber \\
	&& P_{Y|X}^n(y^n|f_n(k_n,\ell_n,s_n,a_n)) \bol{1}[g_n(y^n) \neq (k_n,\ell_n,s_n)], 
\nonumber \\
	\label{eq:bob-decoding-error-probability}
\end{eqnarray}
where $\bol{1}[\cdot]$ is the indicator function. 
Eve's decoder is defined by function $\phi_n:{\cal Z}^n \to {\cal K}_n$ and
the error probability $P_{err}(f_n,\phi_n)$ 
is defined in a similar manner as Eq.~(\ref{eq:bob-decoding-error-probability}).

Let 
\begin{eqnarray*}
P_{\tilde{Z}^n|S_n}(z^n|s_n)
  &=& \sum_{k_n \in {\cal K}_n} \sum_{\ell_n \in {\cal L}_n} 
  	\sum_{a_n \in {\cal A}_n} \frac{1}{|{\cal K}_n||{\cal L}_n||{\cal A}_n|} \\
 &&~	P_{Z|X}^n(z^n|f_n(k_n,\ell_n,s_n,a_n)), \\
P_{\tilde{Z}^n}(z^n) 
  &=& \sum_{s_n \in {\cal S}_n} 
\frac{1}{|{\cal S}_n|}  
P_{\tilde{Z}^n|S_n}(z^n|s_n)
\end{eqnarray*}
be the output distributions of the channel $P_{Z|X}^n$.
In this paper, we consider the security criterion given by
\begin{eqnarray*}
D(f_n) 
&:=& D(P_{S_n \tilde{Z}^n} \| P_{S_n} \times P_{\tilde{Z}^n}) \\
&=& \sum_{s_n \in {\cal S}_n} \frac{1}{|{\cal S}_n|} D(P_{\tilde{Z}^n|S_n}(\cdot|s_n) \| 
P_{\tilde{Z}^n}) \\
&=& I(S_n; \tilde{Z}^n),
\end{eqnarray*}
where $D(\cdot \| \cdot)$ is the divergence, and $I(\cdot; \cdot)$ is the mutual information
\cite{cover}. The coding system investigated in this paper is depicted in Fig.~\ref{Fig:system}.

In this paper, we are interested in the trade-off among
the rate the dummy randomness, and
the rates of the common, private, and confidential messages.
\begin{definition}
The rate quadruple $(R_d,R_0,R_1, R_s)$ is said to be {\em achievable} if
there exists a sequence of Alice's deterministic encoder 
$f_n:{\cal K}_n \times {\cal L}_n \times {\cal S}_n \times {\cal A}_n \to {\cal X}^n$, 
Bob's decoder $g_n:{\cal Y}^n \to {\cal K}_n \times {\cal L}_n \times {\cal S}_n$,
and Eve's decoder $\phi_n:{\cal Z}^n \to {\cal K}_n$ such that
\begin{eqnarray}
\lim_{n \to \infty} P_{err}(f_n,g_n) &=& 0, \\
\lim_{n \to \infty} P_{err}(f_n,\phi_n) &=& 0, \\
\lim_{ n \to \infty} D(f_n) &=& 0, \label{eq:definition-of-achievability-D} \\
\limsup_{n \to \infty} \frac{1}{n} \log |{\cal A}_n| &\le& R_d, \\
\liminf_{n \to \infty} \frac{1}{n} \log |{\cal K}_n| &\ge& R_0, \\
\lim_{n \to \infty} \frac{1}{n} \log |{\cal L}_n| &=& R_1, 
	\label{eq:rate-R1} \\
\liminf_{n \to \infty} \frac{1}{n} \log |{\cal S}_n| &\ge& R_s.
\end{eqnarray}
Then the achievable region ${\cal R}$ is defined 
as the set of all achievable rate quadruples.
\end{definition}

%%% main theorem %%%%
\subsection{Statements of General Results}

The following is our main result in this paper.
\begin{theorem}
\label{theorem:main}
Let ${\cal R}^*$ be a closed convex set consisting of
those quadruples $(R_d,R_0,R_1,R_s)$ for which
there exist auxiliary random variables $(U,V)$ such that
$U \markov V \markov X \markov (Y,Z)$ and
\begin{eqnarray}
R_0 &\le& \min[ I(U;Y), I(U;Z) ], \label{theorem:main-condition1} \\
R_0 + R_1 + R_s &\le& I(V;Y|U) + \min[ I(U;Y), I(U;Z)],  \nonumber \\
 \label{theorem:main-condition2} \\
R_s &\le& I(V;Y|U) - I(V;Z|U), \label{theorem:main-condition4} \\
R_1 + R_d &\ge& I(X;Z|U), \label{theorem:main-condition3} \\
R_d &\ge& I(X; Z|V). \label{theorem:main-condition5}
\end{eqnarray}
Then we have
${\cal R} = {\cal R}^*$.
Moreover, it may be assumed that $V = (U, V^\prime)$ and
that the ranges of $U$ and $V^\prime$ may be assumed to
satisfy $|{\cal U}| \le |{\cal X}| + 3$ and $|{\cal V}^\prime| \le |{\cal X}| + 1$.
\end{theorem}
\begin{proof}
See Section \ref{proof-of-main-result}.
\end{proof}

The conditions on $R_0$ and $R_1 + R_s$ in
Eqs.~(\ref{theorem:main-condition1}) and (\ref{theorem:main-condition2}) resemble
the conditions in the broadcast channel with degraded message sets \cite{korner:77}.
The condition on $R_s$ in Eq.~(\ref{theorem:main-condition4}) exists
because there is a security requirement on the confidential message.
These conditions are exactly the same as those in the conventional BCC
(see Corollary \ref{corollary:bcc}).
The conditions on $R_1$ and $R_d$ in Eqs.~(\ref{theorem:main-condition3}) and 
(\ref{theorem:main-condition5}) 
additionally appear in Theorem \ref{theorem:main} because there are randomness 
constraints in our problem setting.

\begin{remark}
Conventionally, the security requirement defined by
\begin{eqnarray}
\label{eq:weak-security}
\lim_{n \to \infty} \frac{1}{n} D(f_n) = 0
\end{eqnarray}
is usually employed instead of Eq.~(\ref{eq:definition-of-achievability-D}).
Eq.~(\ref{eq:weak-security}) is called weak security criterion
and Eq.~(\ref{eq:definition-of-achievability-D}) is called strong security criterion 
\cite{csiszar:96,maurer:94b}. 
Let $\tilde{{\cal R}}$ be the achievable region in which 
Eq.~(\ref{eq:definition-of-achievability-D})
is replaced by Eq.~(\ref{eq:weak-security}). From the definitions of two regions, ${\cal R} 
\subset \tilde{{\cal R}}$
obviously holds. Actually, we are implicitly showing $\tilde{{\cal R}} \subset {\cal R}^*$ in 
the converse proof
of Theorem \ref{theorem:main}. Thus, ${\cal R} = \tilde{{\cal R}}$. 
\end{remark}

\begin{remark}
As we will find in the achievability proof of the main theorem,
the private message can be used as dummy randomness to
protect the confidential message from Eve. Thus, if we 
define the achievable rate region $\hat{{\cal R}}$ by
replacing Eq.~(\ref{eq:rate-R1}) with 
\begin{eqnarray}
\liminf_{n \to \infty} \frac{1}{n} \log |{\cal L}_n| \ge R_1,
\end{eqnarray}
region $\hat{{\cal R}}$ is broader than region ${\cal R}$.
Indeed, $\hat{{\cal R}}$ is a closed convex set consisting of those
quadruples $(R_d,R_0,R_1,R_s)$ for which there exist auxiliary
random variables $(U,V)$ satisfying the same conditions as
Theorem \ref{theorem:main} except Eq.~(\ref{theorem:main-condition3})\footnote{It
can be proved by just omitting the derivation of Eq.~(\ref{theorem:main-condition3})
in the converse proof of Theorem \ref{theorem:main}.}.
\end{remark}

\begin{remark}
Eq.~(\ref{theorem:main-condition5}) means that there is a certain amount
of dummy randomness that cannot be substituted by the private message.
Note that the difference between the private message and the dummy randomness 
is whether Bob needs to decode it or not.
\end{remark}

Let 
\begin{eqnarray*}
{\cal R}_{\infty} = \{ (R_0,R_1, R_s): \exists R_d \ge 0 ~\mbox{s.t.}~(R_d,R_0,R_1, R_s) \in 
{\cal R} \}
\end{eqnarray*}
be the set of all achievable triplet $(R_0,R_1,R_s)$ by arbitrary stochastic encoder.
By taking sufficiently large $R_d$, we recover the following well known 
result \cite{csiszar:78}\footnote{See also \cite[Theorem 17.13]{csiszar-korner:11}
for the result that does not employ the rate-equivocation formulation.}.
\begin{corollary}
\label{corollary:bcc}
(\cite{csiszar:78})
Region ${\cal R}_\infty$ is a closed convex set consisting of those
triplet $(R_0,R_1,R_s)$ for which there exist auxiliary random variables
$(U,V)$ such that $U \markov V \markov X \markov (Y,Z)$ and
\begin{eqnarray*}
R_0 &\le& \min[I(U;Y), I(U;Z)], \\
R_0 + R_1 + R_s &\le& I(V;Y|U) + \min[ I(U;Y), I(U;Z)], \\
R_s &\le& I(V;Y|U) - I(V;Z|U).
\end{eqnarray*}
\end{corollary}

Let 
\begin{eqnarray*}
{\cal R}_{det} = \{ (R_0,R_1,R_s) :~(0,R_0,R_1,R_s) \in {\cal R} \}
\end{eqnarray*}
be the set of all rate triplets that can be achieved by deterministic 
encoder. This extreme case was solved in \cite{oohama:10},
which can be also derived as a corollary of Theorem \ref{theorem:main}\footnote{In 
\cite{oohama:10},
slightly deferent problem formulation is employed and the achievable region seems slightly 
different from Corollary \ref{corollary:deterministic}. But they are essentially the same.}.
\begin{corollary}
\label{corollary:deterministic}
(\cite{oohama:10})
Let ${\cal R}_{det}^*$ be a closed convex set consisting of
those triplet $(R_0,R_1,R_s)$ for which
there exists an auxiliary random variable $U$ such that
$U \markov X \markov (Y,Z)$ and 
\begin{eqnarray*}
R_0 &\le& \min[ I(U;Y), I(U;Z)], \\
R_0 + R_1 + R_s &\le& I(X;Y|U) + \min[I(U;Y), I(U;Z)], \\
R_s &\le& I(X;Y|U) - I(X;Z|U), \\
R_1 &\ge& I(X;Z|U). \\
\end{eqnarray*}
Then we have ${\cal R}_{det} = {\cal R}_{det}^*$.
\end{corollary}
\begin{proof}
The inclusion ${\cal R}_{det}^* \subset {\cal R}_{det}$ is obvious by
taking $V = X$ in Theorem \ref{theorem:main}. For the opposite inclusion,
note that Eq.~(\ref{theorem:main-condition5}) and $R_d = 0$ imply
\begin{eqnarray*}
I(X;Z|U) = I(V;Z|U) + I(X;Z|V) = I(V;Z|U).
\end{eqnarray*}
We also have $I(V;Y|U) \le I(X;Y|U)$ from the Markov condition
of the auxiliary random variables. Thus, we have
\begin{eqnarray*}
{\cal R}_{det} \subset \{(R_0,R_1,R_s): (0,R_0,R_1,R_s) \in {\cal R}^* \} \subset {\cal 
R}_{det}^*.
\end{eqnarray*}
\end{proof}

Let 
\begin{eqnarray*}
R_d(R_0,R_s) = \inf\{ R_d: (R_d,R_0,0,R_s) \in {\cal R} \}
\end{eqnarray*}
be the infimum rate of dummy randomness needed to achieve the rates $(R_0,R_s)$.
From Theorem \ref{theorem:main}, we can characterize not only the known 
extreme cases (Corollary \ref{corollary:bcc} and Corollary \ref{corollary:deterministic}) 
but also this quantity.
\begin{corollary}
$R_d(R_0,R_s)$ is the optimal solution of the following optimization problem:
\begin{eqnarray*}
\begin{array}{rcl}
\mbox{minimize} & & I(X;Z|U) \\
\mbox{subject to} & & \\
 R_0 &\le& \min[I(U;Y),I(U;Z)], \\
 R_0 + R_s &\le& I(V;Y|U) + \min[I(U;Y),I(U;Z)], \\
 R_s &\le& I(V;Y|U) - I(V;Z|U), 
\end{array}
\end{eqnarray*}
where $(U,V)$ satisfy $U \markov V \markov X \markov (Y,Z)$.
\end{corollary}

Let 
\begin{eqnarray*}
{\cal R}_{ds} = \{ (R_d,R_s):~(R_d,0,0,R_s) \in {\cal R} \}.
\end{eqnarray*}
As a corollary of Theorem \ref{theorem:main}, we also have the following.
\begin{corollary}
\label{corollary-ds}
Let ${\cal R}_{ds}^*$ be a closed convex set consisting of
those rate pair $(R_d,R_s)$ for which
there exist auxiliary random variables $(U,V)$ such that
$U \markov V \markov X \markov (Y,Z)$ and 
\begin{eqnarray}
R_s &\le& I(V;Y|U) - I(V;Z|U), \\
R_d &\ge& I(X;Z|U).
\end{eqnarray}
Then we have ${\cal R}_{ds} = {\cal R}_{ds}^*$.
\end{corollary}
\begin{remark}
\label{remark:size-time-sharing}
The auxiliary random variable $U$ in Corollary \ref{corollary-ds}
only plays a role of time-sharing. Thus, the range of $U$ may
be assumed to satisfy $|{\cal U}| \le 2$.
The same remark is also applied for Corollary \ref{corollary-more-capable}.
\end{remark}

Let 
\begin{eqnarray*}
C_s = \sup\{ R_s:~ \exists R_d \ge 0 \mbox{ s.t. } (R_d,R_s) \in {\cal R}_{ds} \} 
\end{eqnarray*}
be the secrecy capacity, which can be characterized by
the supremum of the rate $R_s$ for which there
exists auxiliary random variable $V$ such that
$V \markov X \markov (Y,Z)$ and
\begin{eqnarray}
\label{eq:secrecy-capacity}
R_s \le I(V; Y) - I(V;Z).
\end{eqnarray}
To achieve the rate given by the right hand side of
Eq.~(\ref{eq:secrecy-capacity}), we conventionally used
the following coding scheme. First, we construct a 
wire-tap channel code for channel pairs $P_{Y|V}$
and $P_{Z|V}$. Then, the code word in ${\cal V}^n$
is transmitted over prefixing channel $P_{X|V}^n$.
If we simulate channel $P_{X|V}^n$ by using
the channel simulation method \cite{steinberg:94},
then we need randomness with rate $H(X|V)$\footnote{We are
implicitly assuming that the empirical distributions of almost 
every code words are close to $P_V$, which is true if we 
use the random coding method.}.
By using this argument, we can derive the following inner bound on ${\cal R}_{ds}$
that can be achieved by combining the ordinary wire-tap channel coding and the 
channel prefixing by the channel simulation method.
\begin{proposition}
Let ${\cal R}_{sim}^*$ be a closed convex set consisting of
those rate pair $(R_d,R_s)$ for which
there exist auxiliary random variables $(U,V)$ such that
$U \markov V \markov X \markov (Y,Z)$ and 
\begin{eqnarray}
R_s &\le& I(V;Y|U) - I(V;Z|U), \\
R_d &\ge& I(V;Z|U) + H(X|V).
\end{eqnarray}
Then we have ${\cal R}_{sim}^*  \subset {\cal R}_{ds}$.
\end{proposition}
Since $I(X;Z|U) = I(V; Z|U) + I(X; Z|V) < I(V;Z|U) + H(X|V)$ in general,
the region ${\cal R}_{ds}^*$ is strictly broader than the region ${\cal R}_{sim}^*$,
i.e., the straightforward combination of the ordinary wire-tap channel coding and
the channel prefixing by the channel simulation is suboptimal.

\begin{corollary}
\label{corollary-more-capable}
Suppose that the channel $P_{Y|X}$ is more capable than
$P_{Z|X}$. Then the region ${\cal R}_{ds} = {\cal R}_{ds}^*$ 
is a closed convex set consisting of
those rate pair $(R_d,R_s)$ for which
there exists an auxiliary random variable $U$ such that
$U \markov X \markov (Y,Z)$ and 
\begin{eqnarray*}
R_s &\le& I(X;Y|U) - I(X;Z|U), \\
R_d &\ge& I(X;Z|U).
\end{eqnarray*}
Moreover, it may be assumed that the ranges of
$U$ may be assumed to satisfy $|{\cal U}| \le 2$.
\end{corollary}
\begin{proof}
See Appendix \ref{proof-of-corollary-main}.
\end{proof}

As we can find from Corollary \ref{corollary-more-capable},
we do not need auxiliary random variable $V$ when the channel 
$P_{Y|X}$ is more capable than $P_{Z|X}$. Thus, two regions
${\cal R}_{ds}^*$ and ${\cal R}^*_{sim}$ coincide.

%%%% Example %%%%%%
\subsection{Numerical Examples}

%%% Example 1 %%%%%%%
First, we consider an example such that ${\cal R}_{ds}^*$ and ${\cal R}^*_{sim}$ coincide.
Suppose that $P_{Y|X}$ and $P_{Z|X}$ are binary symmetric
channels with crossover probabilities $\varepsilon_1$ and $\varepsilon_2$ respectively,
where $\varepsilon_1 < \varepsilon_2$.
In this case, $P_{Z|X}$ is degraded version of $P_{Y|X}$, which also implies that
$P_{Y|X}$ is more capable than $P_{Z|X}$. Thus, we can apply 
Corollary \ref{corollary-more-capable}. Since the auxiliary random variable $U$ only
plays a role of time sharing, region ${\cal R}_{ds}$ is the convex hull of the rates $(R_d,R_s)$ satisfying
\begin{eqnarray*}
R_s &\le& [h(p * \varepsilon_1) - h(\varepsilon_1)] - [h(p * \varepsilon_2) - 
h(\varepsilon_2)], \\
R_d &\ge& h(p * \varepsilon_2) - h(\varepsilon_2)
\end{eqnarray*}
for some input distribution $0 \le P_X(0) = p \le 1$, where $h(\cdot)$ is the binary entropy function
\footnote{Note that
the base of the logarithm is $e$.}
and $x * y = x (1-y) + (1-x) y$ is the binary convolution. 
In Fig.~\ref{Fig:region}, for the case with $\varepsilon_1 = 0.1$ and $\varepsilon_2 = 0.2$ 
respectively, the region ${\cal R}_{ds}$
is plotted. The input distribution achieving $C_s$ is the uniform distribution, and
thus $R_s$ is constant when $R_d \ge \log 2 - h(0.2)$. By using a biased input distribution, 
$R_s$ can be positive
even if $R_d$ is smaller than $\log 2 - h(0.2)$.
%%%%%%%% Fig %%%%%%%%%%%%%%%
\begin{figure}[t]
\centering
\includegraphics[width=\linewidth]{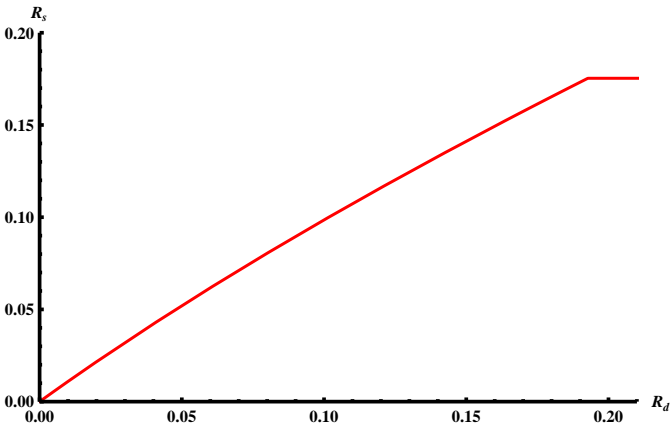}
\caption{The achievable region ${\cal R}_{ds}$ when $P_{Y|X}$ is BSC(0.1) and
$P_{Z|X}$ is BSC(0.2).}
\label{Fig:region}
\end{figure}

%%%%% Example 2 %%%%%%%

Next, we consider an example such that ${\cal R}_{ds}^*$ and ${\cal R}^*_{sim}$ do not 
coincide.
Suppose that $P_{Y|X}$ is a binary symmetric channel with crossover probability $
\varepsilon$
and $P_{Z|X}$ is a binary erasure channel with erasure probability $\delta$.
When $4 \varepsilon (1-\varepsilon) \log 2 < \delta \log 2 < h(\varepsilon)$, it is known 
that
$P_{Y|X}$ is not more capable than $P_{Z|X}$ \cite{nair:10}
and $C_s > 0$\footnote{Actually, for $4 \varepsilon (1-\varepsilon) \log 2 < \delta \log 2 < h(\varepsilon)$,
$P_{Z|X}$ is more capable than $P_{Y|X}$ but $P_{Z|X}$ is not less noisy
than $P_{Y|X}$ \cite{nair:10}. Thus, $I(X;Y) \le I(X;Z)$ for every $P_X$ but there exists
$V$ such that $I(V;Y) > I(V;Z)$, which means that $C_s > 0$ and $V$ is needed to achieve $C_s$.}.
For this example, we can compute the regions ${\cal R}_{ds} = {\cal R}_{ds}^*$
as follows. Since ${\cal R}_{ds}^*$ is a convex set,
for each $R_d$, we can calculate $\max\{R_s : (R_d, R_s) \in {\cal R}_{ds}^* \}$ by
minimizing
\begin{eqnarray}
\label{eq:optimization-bsc-bec}
\max_{P_{UVX}}[ I(V; Y|U) - I(V;Z|U) - \mu(I(X;Z|U) - R_d)]
\end{eqnarray}
with respect to $\mu \ge 0$, where $\mu$ is the slope of the supporting line of ${\cal 
R}_{ds}^*$.
Since $U$ only plays the role of the times sharing in Eq.~(\ref{eq:optimization-bsc-bec}),
we can take $U$ to be constant. Furthermore, by using 
the support lemma \cite{csiszar-korner:11},
we can assume that $|{\cal V}| \le |{\cal X}| = 2$.
Thus, Eq.~(\ref{eq:optimization-bsc-bec}) can be calculated by exhaustive
search of three parameters $P_V(0)$, $P_{X|V}(0|0)$, and $P_{X|V}(1|1)$.
Since $P_V(0) = \frac{1}{2}$ is not necessarily optimal\footnote{When there is no constraint 
on $R_d$, 
it is known that $P_V(0) = \frac{1}{2}$ and $P_{X|V}(0|0) = P_{X|V}(1|1)$ are optimal 
\cite{ozel:11}.} for
$R_d < (1-\delta) \log 2$,  
further reduction of parameters seems difficult.
The region ${\cal R}_{sim}^*$ can be computed in a similar manner.

In Fig.~\ref{Fig:region-bsc-bec}, for the case with $\varepsilon = 0.11$ and $\delta = 
0.45$ respectively,
the region ${\cal R}_{ds} = {\cal R}_{ds}^*$ and ${\cal R}_{sim}^*$ are plotted.
%%%%%%%% Fig %%%%%%%%%%%%%%%
\begin{figure}[t]
\centering
\includegraphics[width=\linewidth]{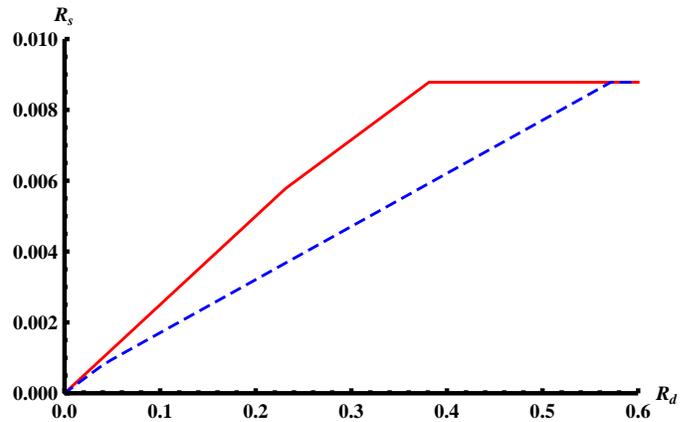}
\caption{The achievable region ${\cal R}_{ds} = {\cal R}^*_{ds}$ (solid line)
and suboptimal inner bound ${\cal R}_{sim}^*$ (dashed line) when $P_{Y|X}$ is BSC(0.11) 
and
$P_{Z|X}$ is BEC(0.45).}
\label{Fig:region-bsc-bec}
\end{figure}

%%%%% Resolvability %%%%%%
\section{Channel Resolvability by Superposition}
\label{section:resolvability}

In this section, we consider the channel resolvability problem.
The result in this section will be utilized in the direct part of
the proof of Theorem \ref{theorem:main}.

In the channel resolvability problem, we want to simulate 
the response $P_Z^n$ by using channel $P_{Z|X}^n$ and
as small number of uniform 
randomness as possible, where $P_{Z}^n$ is the $n$th product of 
\begin{eqnarray*}
P_Z(z) = \sum_x P_X(x) P_{Z|X}(z|x)
\end{eqnarray*}
for input distribution $P_X$.
The simulation is conducted by a deterministic 
map $\varphi_n: {\cal B}_n \to {\cal X}^n$ and uniform random number 
$B_n$ on ${\cal B}_n$. Let 
\begin{eqnarray*}
P_{\tilde{Z}^n}(z^n) = \sum_{b_n \in {\cal B}_n} \frac{1}{|{\cal B}_n|} P_{Z|X}^n(z^n|
\varphi_n(b_n))
\end{eqnarray*} 
be the output distribution with map $\varphi_n$. 
In this paper, the accuracy of the simulation is evaluated by  the divergence criterion
$D(P_{\tilde{Z}^n} \|P^n_{Z})$.
%%%%% known result %%%%%%%%%
It is  well known  \cite{han:93}\footnote{Actually, slightly weaker statement, i.e.,
$D(P_{\tilde{Z}^n} \| P^n_{Z})$ in
Eq.~(\ref{eq:proposition-resolvability-1}) is replaced by
$\frac{1}{n} D(P_{\tilde{Z}^n} \|P^n_{Z})$, was proved in \cite{han:93}.
The present statement can be derived from the result in \cite{hayashi:10}. } that if 
\begin{eqnarray}
\label{eq:resolvability-known-bound}
R > I(X;Z),
\end{eqnarray}
then there exists a sequence of
maps $\{ \varphi_n \}$ satisfying 
\begin{eqnarray}
\label{eq:proposition-resolvability-1}
\lim_{n \to \infty} D(P_{\tilde{Z}^n} \|P^n_{Z}) &=& 0, \\
\label{eq:proposition-resolvability-2}
\limsup_{n \to \infty} \frac{1}{n} \log |{\cal B}_n| &\le& R.
\end{eqnarray}
Typically, a sequence of maps realizing Eqs.~(\ref{eq:proposition-resolvability-1}) and 
(\ref{eq:proposition-resolvability-2})
is constructed by
randomly generating $|{\cal B}_n|$ codeword
$x^n_{1},\ldots, x^n_{|{\cal B}_n|}$ according to $P_{X}^n$.
We denote the generated code ${\cal C}_n$. Then we have the following proposition.
\begin{proposition}
\label{proposition:single-resolvability}
(\cite{hayashi:10})
For every $n \ge 1$, we have
\begin{eqnarray*}
\mathbb{E}_{{\cal C}_n}\left[ D(P_{\tilde{Z}^n} \| P_Z^n) \right]
	\le \frac{1}{\theta |{\cal B}_n|^\theta} e^{n \psi(\theta|P_{Z|X},P_X)}
\end{eqnarray*}
for $0 < \theta \le 1$, where $\mathbb{E}_{{\cal C}_n}[\cdot]$ means taking
the average over the randomly generated code ${\cal C}_n$,
and the function $\psi(\theta|P_{Z|X},P_X)$ is defined by
\begin{eqnarray}
\lefteqn{
\psi(\theta | P_{Z|X},P_X) } \nonumber \\
&=& \log \sum_z \left( \sum_x P_X(x) P_{Z|X}(z|x)^{1+\theta} \right) P_Z(z)^{-\theta}. 
\label{eq:definition-of-psi-2}
\end{eqnarray}
\end{proposition}

In this paper we construct a sequence of maps 
realizing Eq.~(\ref{eq:proposition-resolvability-1})
by a different method. Let $P_{VX}$ be a distribution such that the marginal 
is $P_X$. We first randomly generate $|{\cal M}_{2,n}|$ codeword
$v^n_1,\ldots, v^n_{|{\cal M}_{2,n}|}$ according to the distribution $P_V^n$.
We denote the generated code by ${\cal C}_{2,n}$.
Then for each $1 \le i \le |{\cal M}_{2,n}|$, we randomly generate 
$|{\cal M}_{1,n}|$ codeword $x^n_{i1},\ldots, x^n_{i |{\cal M}_{1,n}|}$
according to the distribution $P_{X|V}^n(\cdot | v^n_i)$.
We denote the generated code by ${\cal C}_{1,n}$.
The empirical distribution of the codeword is given by
\begin{eqnarray*}
P_{\tilde{V}^n \tilde{X}^n}(v^n,x^n) 
&=& \sum_{i \in {\cal M}_{2,n}} \sum_{j \in {\cal M}_{1,n}}  \frac{1}{|{\cal M}_{2,n}| |{\cal 
M}_{1,n}|} \\
&&~~~   \bol{1}[v^n_i = v^n, x^n_{ij} = x^n], \\
P_{\tilde{V}^n}(v^n)
&=& \sum_{i \in {\cal M}_{2,n}} \frac{1}{|{\cal M}_{2,n}|} \bol{1}[v^n_i = v^n], \\
P_{\tilde{X}^n|\tilde{V}^n}(x^n|v^n) 
&=& \frac{P_{\tilde{V}^n \tilde{X}^n}(v^n,x^n) }{P_{\tilde{V}^n}(v^n)},    
\end{eqnarray*}
and the output distribution is given by
\begin{eqnarray*}
P_{\tilde{Z}^n}(z^n) = \sum_{v^n, x^n} P_{\tilde{V}^n \tilde{X}^n}(v^n, x^n) 
P_{Z|X}^n(z^n|x^n).
\end{eqnarray*}
For this construction, we have the following lemma.
%%%% Theorem %%%%%%%%
\begin{lemma}
\label{theorem:resolvability}
For every $n \ge 1$, we have
\begin{eqnarray}
\lefteqn{ \mathbb{E}_{{\cal C}_{1,n} {\cal C}_{2,n}} \left[ D(P_{\tilde{Z}^n} \| P_Z^n) \right] } 
\nonumber \\
 &\le& \frac{1}{\theta |{\cal M}_{1,n}|^{\theta}} e^{n \psi(\theta| P_{Z|X},P_{X|V},P_V)} 
\nonumber \\
 && + \frac{1}{\theta^\prime |{\cal M}_{2,n}|^{\theta^\prime}} e^{n \psi(\theta^\prime| 
P_{Z|V},P_V)}
  \label{eq:resolvability-bound-divergence}
\end{eqnarray}
for $0 < \theta, \theta^\prime \le 1$,
where $\mathbb{E}_{{\cal C}_{1,n} {\cal C}_{2,n}}[ \cdot ]$ means taking
the average over the randomly generated codes ${\cal C}_{1,n}$ and ${\cal C}_{2,n}$,
the function $\psi(\theta | P_{Z|X},P_{X|V},P_V)$ is defined as
\begin{eqnarray}
\lefteqn{
\psi(\theta | P_{Z|X},P_{X|V},P_V) } \nonumber \\
&=& \log \sum_v P_V(v) 
 \sum_z \nonumber \\
 && \left( \sum_x P_{X|V}(x|v) P_{Z|X}(z|x)^{1+\theta} \right) P_{Z|V}(z|v)^{- \theta}, 
\nonumber \\
 \label{eq:definition-of-psi}
\end{eqnarray}
and
$\psi(\theta | \cdot,\cdot)$ is defined in Eq.~(\ref{eq:definition-of-psi-2}).
\end{lemma}
\begin{proof}
See Appendix \ref{proof-of-theorem:resolvability}.
\end{proof}
%%%% Corollary %%%%%%
\begin{corollary}
\label{corollary:resolvability}
If $R_1 > I(X; Z|V)$ and $R_2 > I(V;Z)$, there exists a sequence of
map $\varphi_n:{\cal M}_{1,n} \times {\cal M}_{2,n} \to {\cal X}^n$ such that
\begin{eqnarray}
\label{eq:corollary-resolvability-2}
\lim_{n \to \infty} D(P_{\tilde{Z}^n} \|P^n_{Z}) &=& 0, \\
\label{eq:corollary-resolvability-3}
\limsup_{n \to \infty} \frac{1}{n} \log |{\cal M}_{1,n}| &\le& R_1, \\
\label{eq:corollary-resolvability-4}
\limsup_{n \to \infty} \frac{1}{n} \log |{\cal M}_{2,n}| &\le& R_2.
\end{eqnarray}
\end{corollary}
\begin{proof}
See Appendix \ref{proof-of-corollary:resolvability}.
\end{proof}

From Corollary \ref{corollary:resolvability}, we find that the channel resolvability
coding scheme proposed in this section can achieve the rate shown
in Eq.~(\ref{eq:resolvability-known-bound}), i.e., $I(X; Z) = I(V;Z) + I(X;Z|V)$.
Splitting the randomness into two part does not have any meaning in
the channel resolvability coding, but as we will find in Section 
\ref{proof:direct-theorem-main},
this coding scheme does have meaning when we send the confidential message.

%%%%%
\section{Proofs of Main Results}
\label{proof-of-main-result}

\subsection{Proof of Direct Part of Theorem \ref{theorem:main}}
\label{proof:direct-theorem-main}

We prove the direct part of Theorem \ref{theorem:main} by using the
result in Section \ref{section:resolvability}. 
The direct part of the theorem follows from the following 
Lemma \ref{lemma:inner-bound} and Lemma \ref{lemma:rate-spliting}.
\begin{lemma}
\label{lemma:inner-bound}
Let ${\cal R}^{(in)}$ be a closed convex set consisting of 
those quadruples $(R_d,R_0,R_1,R_s)$ for which there
exist auxiliary random variables $(U,V)$ such that
$U \markov V \markov X \markov (Y,Z)$ and
\begin{eqnarray}
R_0 &\le& I(U;Z), \label{eq:inner-1} \\
R_1 + R_s &\le& I(V;Y|U), \label{eq:inner-2} \\
R_0 + R_1 + R_s &\le& I(V; Y), \label{eq:inner-3} \\
R_1  &\ge& I(V; Z|U), \label{eq:inner-4} \\
%R_s &\le& I(V; Y|U) - I(V;Z|U), \label{eq:inner-5} \\
R_d &\ge& I(X;Z|V). \label{eq:inner-6}
\end{eqnarray}
Then ${\cal R}^{(in)} \subset {\cal R}$.
\end{lemma}

We note the following observation.
From the definition of the problem, if 
\begin{eqnarray*}
(R_d - r_d ,R_0 + r_0, R_1 - r_0 - r_s + r_d, R_s + r_s) \in {\cal R}
\end{eqnarray*} 
for some $r_d,r_0,r_s \ge 0$,
then we also have 
$(R_d,R_0,R_1,R_s) \in {\cal R}$. By using this argument, we have the following.
\begin{lemma}
\label{lemma:rate-spliting}
We have ${\cal R}^* \subset {\cal R}$.
\end{lemma}
\begin{proof}
See Appendix \ref{appendix:rate-spliting}.
\end{proof}

We now prove Lemma \ref{lemma:inner-bound}.
For a while, we consider the case with $n=1$ and omit the superscript
and subscript to simplify the notation. 
%We use two kind of dummy messages $i \in {\cal M}_2$ and
%$j \in {\cal M}_1$, i.e., ${\cal A}_n = {\cal M}_1 \times {\cal M}_2$.
For each common message $k \in {\cal K}$, we randomly generate
codeword $u_k$ according to distribution $P_U$. We denote such a
code ${\cal C}_0$. For each $k$ and for each 
$(\ell,s) \in {\cal L} \times {\cal S}$, we randomly
generate codeword $v_{k \ell s}$ according to distribution
$P_{V|U}(\cdot| u_k)$. We denote such a code ${\cal C}_2$.
For each $(k,\ell,s)$ and for each $a \in {\cal A}$, we randomly
generate codeword $x_{k \ell s a}$ according to distribution
$P_{X|V}(\cdot | v_{k \ell s})$. We denote such a code ${\cal C}_1$.

For real numbers $\alpha_0, \alpha_1,\alpha_2 \ge 0$ specified later, let 
\begin{eqnarray*}
{\cal T}_0 &=& \{ (u,z): P_{Z|U}(z|u) \ge e^{\alpha_0} P_Z(z) \}, \\
{\cal T}_1 &=& \{ (u,v,y): P_{Y|V}(y|v) \ge e^{\alpha_1} P_{Y|U}(y|u) \}, \\
{\cal T}_2 &=& \{ (v,y): P_{Y|V}(y|v) \ge e^{\alpha_2} P_Y(y) \}, 
\end{eqnarray*}
and let ${\cal T} = {\cal T}_1 \cap ({\cal U} \times {\cal T}_2)$.
Eve's decoding region is defined by
\begin{eqnarray*}
{\cal D}_k = \{ z: (u_k,z) \in {\cal T}_{0}, (u_{\hat{k}},z) \notin {\cal T}_{0} ~\forall \hat{k} 
\neq k \},
\end{eqnarray*}
i.e., $\phi(z) = k$ if $z \in {\cal D}_k$\footnote{If $z \notin {\cal D}_k$ for every $k \in 
{\cal K}$, 
we set $\phi(z) = 1$, which is not important in our analysis of error probability. A similar
remark is also applied for Bob's decoder.}. 
Bob only decode $(k,\ell,s)$ and he does not decode dummy randomness $a \in {\cal A}$.
Bob's decoding region is defined by 
\begin{eqnarray*}
\lefteqn{ {\cal D}_{k \ell s} = \{ y: (u_k,v_{k \ell s},y) \in {\cal T}, } \\
	&& (u_{\hat{k}},v_{\hat{k} \hat{\ell} \hat{s}}, y) \notin {\cal T}~
	\forall
	(\hat{k},\hat{\ell},\hat{s}) \neq (k,\ell,s) \},
\end{eqnarray*}
i.e., $g(y) = (k,\ell,s)$ if $y \in {\cal D}_{k \ell s}$.

By the above code construction, we have the following.
\begin{lemma}
\label{lemma:error-analysis}
We have
\begin{eqnarray}
\lefteqn{ \mathbb{E}_{{\cal C}_0 {\cal C}_1 {\cal C}_2} [P_{err}(f,g)] } \nonumber \\
	&\le& P_{UVY}({\cal T}_1^c) + P_{VY}({\cal T}_2^c) \nonumber \\
	&& + |{\cal L}| |{\cal S}|  e^{- \alpha_1} + 
		|{\cal K}| |{\cal L}| |{\cal S}|  e^{- \alpha_2}, \label{eq:error-analysis-1} \\
\lefteqn{ \mathbb{E}_{{\cal C}_0 {\cal C}_1 {\cal C}_2} [P_{err}(f,\phi)] } \nonumber \\
	&\le& P_{UZ}({\cal T}_0^c) + |{\cal K}| e^{- \alpha_0} \label{eq:error-analysis-2}
\end{eqnarray}
and
\begin{eqnarray}
\mathbb{E}_{{\cal C}_0 {\cal C}_1 {\cal C}_2} [ D(f)] 
	&\le& \frac{1}{\theta |{\cal A}|^\theta} e^{\psi(\theta|P_{Z|X},P_{X|V},P_{V})} 
\nonumber \\
	&& +\frac{1}{\theta^\prime |{\cal L}|^{\theta^\prime}} e^{\psi(\theta^
\prime|P_{Z|V},P_{V|U},P_U)}
	\label{eq:error-analysis-3}
\end{eqnarray}
for $0 < \theta, \theta^\prime \le 1$,
where the functions $\psi(\theta|\cdot,\cdot,\cdot)$ is defined by
Eq.~(\ref{eq:definition-of-psi}).
\end{lemma}
\begin{proof}
See Appendix \ref{proof-lemma:error-analysis}.
\end{proof}

We apply Lemma \ref{lemma:error-analysis} for asymptotic case.
For $(R_d, R_0,R_1,R_d) \in {\cal R}^{(in)}$ and arbitrary small $\delta >0$,
we set $|{\cal K}_n| = \lfloor e^{n (R_0 - 2 \delta)} \rfloor$,
$|{\cal L}_n| = \lfloor e^{n(R_1 + 2 \delta)} \rfloor$,
$|{\cal S}_n| = \lfloor e^{n (R_s - 4 \delta)} \rfloor$,
$|{\cal A}_n| = \lfloor e^{n(R_d + 2 \delta)} \rfloor$,
$\alpha_0 = I(U; Z) - \delta$,
$\alpha_1 = I(V;Y|U) - \delta$,
$\alpha_2 = I(V; Y) - \delta$.
Then, 
\begin{eqnarray*}
|{\cal L}_n| |{\cal S}_n| e^{-  \alpha_1 n} 
	&\le& e^{- n( I(V;Y|U) - R_1 - R_s  + \delta) }, \\
|{\cal K}_n| |{\cal L}_n| |{\cal S}_n| e^{- \alpha_2 n} &\le& e^{- n( I(V;Y) - R_0 - R_1 - R_s  + 
3 \delta)}, \\
|{\cal K}_n| e^{- \alpha_0 n} &\le& e^{- n( I(U;Z) - R_0 + \delta)}
\end{eqnarray*}
converge to $0$ asymptotically. Furthermore, by  the law of large numbers,
$P_{UVY}^n({\cal T}_{1,n}^c)$, $P_{VY}^n({\cal T}_{2,n}^c)$,
and $P_{UZ}^n({\cal T}_{0,n}^c)$ also converge to $0$ asymptotically.

Since $\psi^\prime(0|P_{Z|X},P_{X|V},P_V) = I(X;Z|V)$, there exists
$\theta_0 > 0$ such that
\begin{eqnarray*}
\frac{\psi(\theta_0|P_{Z|X},P_{X|V},P_V)}{\theta_0} \le I(X;Z|V) + \delta \le R_d + \delta,
\end{eqnarray*}
which implies
\begin{eqnarray*}
- \frac{\theta_0}{n} \log |{\cal A}_n| + \psi(\theta_0|P_{Z|X},P_{X|V},P_V) \le - \delta.
\end{eqnarray*}
Thus, 
\begin{eqnarray*}
\frac{1}{\theta_0 |{\cal A}_n|^{\theta_0}} e^{n \psi(\theta_0|P_{Z|X},P_{X|V},P_V)}
\end{eqnarray*}
exponentially converges to $0$ asymptotically.
Similarly, since
$\psi^\prime(0|P_{Z|V},P_{V|U},P_{U}) = I(V;Z|U)$, there exists 
$\theta_0^\prime > 0$ such that
\begin{eqnarray*}
\frac{\psi(\theta_0^\prime|P_{X|V},P_{V|U},P_U)}{\theta_0^\prime}
	\le I(V;Z|U) + \delta \le R_1 + \delta,
\end{eqnarray*}
which implies 
\begin{eqnarray*}
- \frac{\theta_0^\prime}{n} \log |{\cal L}_n|  + 
	\psi(\theta_0^\prime|P_{X|V},P_{V|U},P_U) \le - \delta.
\end{eqnarray*}
Thus, 
\begin{eqnarray*}
\frac{1}{\theta_0^\prime |{\cal L}_n|^{\theta^\prime_0}} e^{n \psi(\theta_0^
\prime|P_{X|V},P_{V|U},P_U)}
\end{eqnarray*}
exponentially converges $0$ asymptotically.
This completes the proof of Lemma \ref{lemma:inner-bound}. \qed

%%%%% Converse %%%%%%%
\subsection{Proof of Converse Part of Theorem \ref{theorem:main}}

Suppose that $(R_d, R_0, R_1, R_s) \in {\cal R}$.
Then, for arbitrary $\gamma > 0$, there exists $n$ such that
\begin{eqnarray*}
n(R_0 - \gamma) &\le& \log |{\cal K}_n|, \\
n(R_0 + R_1 + R_s - \gamma) &\le& \log |{\cal K}_n| |{\cal L}_n| |{\cal S}_n|, \\
n(R_s - \gamma) &\le& \log |{\cal S}_n|, \\
n(R_1 + R_d + \gamma) &\ge& \log |{\cal L}_n| |{\cal A}_n|, \\
n(R_d + \gamma) &\ge& \log |{\cal A}_n|. 
\end{eqnarray*}
By combining these inequalities with the following
Lemma \ref{lemma:converse-fano} and Lemma \ref{lemma:converse-single-letter}, we have 
the converse part
of the theorem. The statement about the range sizes of $U$ and $V$
can be proved exactly in the same manner as \cite[Theorem 17.13]{csiszar-korner:11}.
It should be noted that 
Eqs.~(\ref{theorem:main-condition1})--(\ref{theorem:main-condition4})
are derived in the same manner as \cite[Theorem 17.13]{csiszar-korner:11} and the 
construction of
the auxiliary random variables are also the same.
Eqs.~(\ref{theorem:main-condition3}) and (\ref{theorem:main-condition5}) are additionally 
proved 
in this paper by using the fact that the encoder is deterministic given the dummy 
randomness.

\begin{lemma}
\label{lemma:converse-fano}
There exists $\varepsilon_n \to 0$ such that
\begin{eqnarray*}
\lefteqn{\log |{\cal K}_n| } \\ &\le& I(K_n; Y^n) + n \varepsilon_n, \\
\lefteqn{\log |{\cal K}_n| } \\ &\le& I(K_n; Z^n) + n \varepsilon_n, \\
\lefteqn{\log |{\cal K}_n| |{\cal L}_n| |{\cal S}_n| } \\
	&\le& I(K_n,L_n,S_n; Y^n) + n \varepsilon_n, \\
\lefteqn{ \log |{\cal K}_n| |{\cal L}_n| |{\cal S}_n| } \\
	&\le& I(L_n,S_n; Y^n|K_n) + I(K_n; Z^n) + 2 n \varepsilon_n, \\
\lefteqn{ \log |{\cal S}_n| } \\ 
	&\le& I(L_n,S_n; Y^n|K_n) - I(L_n,S_n;Z^n|K_n) + 4 n \varepsilon_n, \\
\lefteqn{ \log |{\cal L}_n| |{\cal A}_n| } \\ &\ge& 
	I(X^n; Z^n|K_n) - 2 n \varepsilon_n, \\
\lefteqn{ \log |{\cal A}_n| } \\
	&\ge& I(X^n; Z^n| K_n, L_n,S_n). 
\end{eqnarray*}
\end{lemma}
\begin{proof}
By using Fano's inequality, we have
\begin{eqnarray*}
\log |{\cal K}_n| 
	&=& H(K_n) \\
	&=& I(K_n; Y^n) + H(K_n|Y^n) \\
	&\le& I(K_n; Y^n) + n \varepsilon_n, \\
\end{eqnarray*}
and
\begin{eqnarray*}
\log |{\cal K}_n| \le I(K_n; Z^n) + n \varepsilon_n.
\end{eqnarray*}
By using Fano's inequality, we also have
\begin{eqnarray*}
\log |{\cal K}_n| |{\cal L}_n| |{\cal S}_n| 
	&=& H(K_n,L_n,S_n) \\
	&\le& I(K_n,L_n,S_n; Y^n) + n \varepsilon_n
\end{eqnarray*}
and
\begin{eqnarray*}
\lefteqn{ \log |{\cal K}_n| |{\cal L}_n| |{\cal S}_n|  } \\
	&=& H(L_n, S_n|K_n) + H(K_n) \\
	&=& I(L_n, S_n; Y^n| K_n) + I(K_n; Z^n) + 2 n \varepsilon_n.
\end{eqnarray*}
By using the security condition and Fano's inequality, we have
\begin{eqnarray}
\lefteqn{ I(S_n; Z^n|K_n) } \nonumber \\
	&=& I(S_n,K_n ; Z^n) - I(K_n;Z^n) \nonumber \\
	&=& I(S_n; Z^n) + I(K_n;Z^n|S_n) - I(K_n;Z^n) \nonumber \\
	&\le& I(S_n;Z^n) + H(K_n|Z^n) \nonumber \\
	&\le& 2 n \varepsilon_n.
	\label{eq:szk-bound}
\end{eqnarray}
By using Fano's inequality and
by using Eq.~(\ref{eq:szk-bound}), we have
\begin{eqnarray*}
\lefteqn{ \log |{\cal S}_n| } \\
	&=& H(S_n|K_n) \\
	&\le& I(S_n; Y^n|K_n) + n \varepsilon_n \\
	&=& I(L_n,S_n; Y^n|K_n) - I(L_n;Y^n|S_n,K_n) + n \varepsilon_n \\
	&\le& I(L_n,S_n; Y^n|K_n) - H(L_n|S_n,K_n) + 2 n \varepsilon_n \\
	&\le& I(L_n,S_n; Y^n|K_n) - I(S_n;Z^n|K_n) \\
	&& - H(L_n|S_n,K_n) + 4 n \varepsilon_n \\ 
	&\le& I(L_n,S_n; Y^n|K_n) - I(L_n,S_n;Z^n|K_n) + 4 n \varepsilon_n.  
\end{eqnarray*}

By noting that $f_n$ is a deterministic function
and by using Eq.~(\ref{eq:szk-bound}), we have
\begin{eqnarray*}
\lefteqn{ \log |{\cal L}_n| |{\cal A}_n| } \\
	&\ge& H(X^n|K_n,S_n) \\
	&\ge& I(X^n; Z^n| K_n, S_n) \\
	&=& I(X^n, S_n ; Z^n|K_n) - I(S_n; Z^n|K_n) \\
	&\ge& I(X^n; Z^n|K_n) - 2 n \varepsilon_n.
\end{eqnarray*}
Finally, by noting that $f_n$ is a deterministic function, we have
\begin{eqnarray*}
\log |{\cal A}_n| 
	&\ge& H(X^n|K_n,L_n,S_n) \\
	&\ge& I(X^n; Z^n|K_n, L_n, S_n).
\end{eqnarray*}
\end{proof}

\begin{lemma}
\label{lemma:converse-single-letter}
For fixed $n$, let $T$ be the random variable that 
is uniformly distributed over $\{1,\ldots,n\}$
and independent of the other random variables.
Define the following auxiliary random variables
\begin{eqnarray*}
U_t &=& (K_n,Y_1^{t-1},Z_{t+1}^n), \\
V_t &=& (L_n,S_n,U_t), \\
U &=& (U_T,T), \\
V &=& (V_T, T).
\end{eqnarray*}
Then, we have
\begin{eqnarray}
\lefteqn{ I(K_n;Y^n) } \nonumber \\
	&\le& n I(U; Y), \label{single-letter-1} \\
\lefteqn{ I(K_n;Z^n) } \nonumber \\
	&\le& n I(U; Z), \label{single-letter-2} \\
\lefteqn{ I(K_n,L_n,S_n;Y^n) } \nonumber \\ 
	&\le& n[ I(V;Y|U) + I(U;Y)], \label{single-letter-3} \\
\lefteqn{ I(L_n,S_n;Y^n|K_n) + I(K_n; Z^n) } \nonumber \\
	&\le& n[I(V;Y|U) + I(U;Z)], \label{single-letter-4} \\
\lefteqn{ I(L_n,S_n;Y^n|K_n) - I(L_n,S_n;Z^n|K_n) } \nonumber \\
	&\le& n[I(V;Y|U) - I(V;Z|U)], \label{single-letter-6} \\
\lefteqn{ I(X^n;Z^n|K_n) } \nonumber \\
	&\ge& n I(X;Z|U), \label{single-letter-5} \\
\lefteqn{ I(X^n;Z^n|K_n,L_n,S_n) } \nonumber \\
	&\ge& n I(X;Z|V) \label{single-letter-7}
\end{eqnarray}
and 
\begin{eqnarray}
\label{eq:markov-proof}
U \markov V \markov X \markov (Y,Z).
\end{eqnarray}
\end{lemma}
\begin{proof}
Since the proof of Eq.~(\ref{eq:markov-proof}) is well known \cite{csiszar-korner:11}, we 
only
prove the other inequalities.
\paragraph{Proof of Eq.~(\ref{single-letter-1})}
\begin{eqnarray*}
\lefteqn{ I(K_n; Y^n) } \\
	&=& \sum_{t=1}^n I(K_n; Y_t|Y_1^{t-1}) \\
	&\le& \sum_{t=1}^n I(K_n,Y_1^{t-1},Z_{t+1}^n; Y_t) \\
	&=& \sum_{t=1}^n I(U_t; Y_t) \\
	&=& n I(U_T; Y_T|T) \\
	&=& n I(U_T,T; Y_T) \\
	&=& n I(U; Y).
\end{eqnarray*}

\paragraph{Proof of Eq.~(\ref{single-letter-2})}
\begin{eqnarray*}
\lefteqn{ I(K_n; Z^n) } \\
	&=& \sum_{t=1}^n I(K_n; Z_t|Z_{t+1}^n) \\
	&\le& \sum_{t=1}^n I(K_n,Y_1^{t-1},Z_{t+1}^n; Z_t) \\
	&=& \sum_{t=1}^n I(U_t; Z_t) \\
	&=& n I(U_T; Z_T|T) \\
	&=& n I(U_T,T; Z_T) \\
	&=& n I(U; Z).
\end{eqnarray*}

\paragraph{Proof of Eq.~(\ref{single-letter-3})}
\begin{eqnarray*}
\lefteqn{ I(K_n,L_n,S_n; Y^n) } \\
	&=& \sum_{t=1}^n I(K_n,L_n,S_n;Y_t| Y_1^{t-1}) \\
	&\le& \sum_{t=1}^n I(K_n,L_n,S_n,Y_1^{t-1},Z_{t+1}^n; Y_t) \\
	&=& \sum_{t=1}^n I(U_t,V_t; Y_t) \\
	&=& n I(U_T,V_T; Y_T|T) \\
	&=& n I(U,V; Y) \\
	&=& n[ I(V;Y|U) + I(U;Y)].
\end{eqnarray*}

\paragraph{Proof of Eq.~(\ref{single-letter-4})}
\begin{eqnarray*}
\lefteqn{ I(L_n,S_n; Y^n|K_n) + I(K_n; Z^n) } \\
	&=& \sum_{t=1}^n [ I(L_n,S_n; Y_t|K_n, Y_1^{t-1}) + I(K_n; Z_t|Z_{t+1}^n) ] \\
	&\le& \sum_{t=1}^n[ I(L_n,S_n; Y_t|K_n,Y_1^{t-1},Z_{t+1}^n)  \\
	&& + I(Z_{t+1}^n; Y_t|K_n,Y_1^{t-1}) + I(K_n; Z_T|Z_{t+1}^n) ] \\
	&\stackrel{(a)}{=}& \sum_{t=1}^n [ I(L_n,S_n;Y_t|K_n,Y_t^{t-1},Z_{t+1}^n) \\
	&& + I(Y_1^{t-1}; Z_t|K_n,Z_{t+1}^n) + I(K_n;Z_t|Z_{t+1}^n) ]\\
	&\le& \sum_{t=1}^n [I(L_n,S_n; Y_t|K_n,Y_1^{t-1},Z_{t+1}^n) \\
	&& + I(K_n,Y_1^{t-1},Z_{t+1}^n; Z_t) ] \\
	&=& \sum_{t=1}^n [ I(V_t; Y_t|U_t) + I(U_t; Z_t) ] \\
	&=& n [I(V_T;Y_T|U_T,T) + I(U_T; Z_T|T) ] \\
	&=& n [ I(V; Y|U) + I(U; Z)],
\end{eqnarray*}
where we used Csisz\'ar's sum identity \cite{elgamal-kim-book} in (a).

\paragraph{Proof of Eq.~(\ref{single-letter-6})}
\begin{eqnarray*}
\lefteqn{ I(L_n,S_n; Y^n|K_n) - I(L_n,S_n; Z^n|K_n) } \\
	&=& \sum_{t=1}^n [ I(L_n,S_n; Y_t|K_n,Y_1^{t-1} ) - I(L_n,S_n; Z_t|K_n,Z_{t+1}^n) ] \\
	&\stackrel{(a)}{=}& \sum_{t=1}^n [ I(L_n,S_n; Y_t|K_n,Y_t^{t-1} ) \\
	&& + I(Z_{t+1}^n; Y_t|K_n, L_n,S_n,Y_1^{t-1}) \\
	&& - I(Y_1^{t-1}; Z_t|K_n,L_n,S_n,Z_{t+1}^n) \\
	&& - I(L_n,S_n; Z_t|K_n,Z_{t+1}^n) ] \\
	&=& \sum_{t=1}^n [ I(L_n,S_n,Z_{t+1}^n; Y_t|K_n,Y_1^{t-1})  \\
	&& - I(L_n,S_n,Y_1^{t-1}; Z_t|K_n,Z_{t+1}^n) ] \\
	&\stackrel{(b)}{=}& \sum_{t=1}^n[ I(L_n,S_n; Y_t|K_n,Y_1^{t-1},Z_{t+1}^n) \\
	&& + I(Z_{t+1}^n; Y_t|K_n,Y_1^{t-1}) - I(Y_1^{t-1}; Z_t|K_n,Z_{t+1}^n) \\
	&& - I(L_n,S_n; Z_t|K_n,Y_1^{t-1},Z_{t+1}^n) ] \\
	&=& \sum_{t=1}^n[ I(L_n,S_n;Y_t|K_n,Y_1^{t-1},Z_{t+1}^n) \\
	&& - I(L_n,S_n;Z_t|K_n,Y_1^{t-1},Z_{t+1}^n)] \\
	&=& \sum_{t=1}^n [ I(V_t; Y_t|U_t) - I(V_t; Z_t|U_t) ] \\
	&=& n [ I(V_T; Y_T|U_T,T) - I(V_T; Z_T|U_T,T) ]\\
	&=& n [ I(V; Y|U) - I(V; Z|U) ],
\end{eqnarray*}
where (a) and (b) follow from Csisz\'ar's sum identity \cite{elgamal-kim-book}.

\paragraph{Proof of Eq.~(\ref{single-letter-5})}
\begin{eqnarray*}
\lefteqn{ I(X^n;Z^n|K_n) } \\
	&=& \sum_{t=1}^n[ H(Z_t|K_n,Z_{t+1}^n) - H(Z_t|K_n,X^n,Z_{t+1}^n) ] \\
	&\stackrel{(a)}{\ge}& \sum_{t=1}^n[ H(Z_t|K_n,Y_1^{t-1},Z_{t+1}^n) - 
H(Z_t|K_n,X_t,Y_1^{t-1},Z_{t+1}^n) ] \\
	&=& \sum_{t=1}^n I(X_t; Z_t|K_n,Y_1^{t-1},Z_{t+1}^n) \\
	&=& \sum_{t=1}^n I(X_t ; Z_t|U_t) \\
	&=& n I(X_T; Z_T|U_T,T) \\
	&=& n I(X; Z|U),
\end{eqnarray*}
where (a) follows from the fact that $(K_n,X_1^{t-1},X_{t+1}^n,Y_1^{t-1},Z_{t+1})$, $X_t$, 
and 
$Z_t$ form Markov chain.

\paragraph{Proof of Eq.~(\ref{single-letter-7})}
\begin{eqnarray*}
\lefteqn{ I(X^n; Z^n|K_n,L_n,S_n) } \\
	&=& \sum_{t=1}^n[ H(Z_t|K_n,L_n,S_n,Z_{t+1}^n) \\
	&& - H(Z_t|K_n,L_n,S_n,X^n,Z_{t+1}^n) ] \\
	&\stackrel{(a)}{\ge}& \sum_{t=1}^n[ H(Z_t|K_n,L_n,S_n,Y_1^{t-1},Z_{t+1}^n)\\
	&&	- H(Z_t|K_n,L_n,S_n,X_t,Y_1^{t-1},Z_{t+1}^n) ] \\
	&=& \sum_{t=1}^n I(X_t; Z_t|K_n,L_n,S_n,Y_1^{t-1},Z_{t+1}^n) \\
	&=& \sum_{t=1}^n I(X_t; Z_t|V_t) \\
	&=& n I(X_T; Z_T|V_T,T) \\
	&=& n I(X; Z| V),
\end{eqnarray*}
where (a) follows from the fact that
$(K_n$, $L_n$, $S_n$, $X_1^{t-1}$, $X_{t+1}^n$, $Y_1^{t-1}$, $Z_{t+1}^n)$,
$X_t$, and $Z_t$ form Markov chain.
\end{proof}

%%%%%%
\section{Conclusion}
\label{section:conclusion}

In this paper, we investigated the trade-off between the
rate of the random number, the rates of  common, private,
and confidential messages. 

As a by-product of our result,
Lemma \ref{theorem:resolvability} can be also applied 
to the three receiver wire-tap channel, and
the lower bound of secrecy capacity obtained
in \cite[Corollary 1]{chia:12} with strong security can be proved.
%Showing the strong security of \cite[Theorem 1]{chia:09}, which involves Marton's
%coding \cite{marton:79}, is a future research agenda.

%%% Ack %%%%%%%%%
\section*{Acknowledgment}

This research was initiated by a discussion with Prof.~Ryutaroh Matsumoto
about the deterministic encoding result in \cite{oohama:10}.
The authors would like to thank him
for bringing the authors' attention to the randomness constrained
stochastic encoding problem.
This research is partly supported by
Grand-in-Aid for Young Scientists(B):2376033700,
Grant-in-Aid for Scientific Research(B):2336017202,
and Grant-in-Aid for Scientific Research(A):2324607101.

%%%% Appendix %%%%%%
\appendix

\subsection{Proof of Lemma \ref{theorem:resolvability}}
\label{proof-of-theorem:resolvability}

For simplicity of notation, we only prove the statement for $n=1$,
and the subscript $n$ is omitted in the proof.
The statement for $n \ge 2$ can be proved by regarding the
$n$th product distribution as one distribution and by noting that
\begin{eqnarray*}
\psi(\theta| P_{Z|X}^n,P_{X|V}^n,P_V^n) = n \psi(\theta|P_{Z|X},P_{X|V},P_V)
\end{eqnarray*}
and
\begin{eqnarray*}
\psi(\theta^\prime| P_{Z|V}^n, P_V^n) = n \psi(\theta^\prime| P_{Z|V},P_V)
\end{eqnarray*}
hold.

We first note the following observations.
By taking average over the randomly generated codes ${\cal C}_1$ and
${\cal C}_2$, we have
\begin{eqnarray*}
\lefteqn{
\mathbb{E}_{{\cal C}_1 {\cal C}_2} \left[ P_{\tilde{V}}(v) P_{\tilde{X}|\tilde{V}}(x|v) \right]
} \\
&=& \mathbb{E}_{{\cal C}_1 {\cal C}_2} \left[ \sum_{i \in {\cal M}_2} \sum_{j \in {\cal M}_1} 
\frac{1}{|{\cal M}_1| |{\cal  M}_2|} 
	\bol{1}[v_i = v, x_{ij} = x] \right] \\
&=& P_{VX}(v,x) 
\end{eqnarray*}
and
\begin{eqnarray}
\mathbb{E}_{{\cal C}_1 {\cal C}_2} \left[ P_{\tilde{V}}(v) \right]
 &=& \mathbb{E}_{{\cal C}_2} \left[ P_{\tilde{V}}(v)  \right] \nonumber \\
 &=& \mathbb{E}_{{\cal C}_2 }\left[ \sum_{i \in {\cal M}_2} \frac{1}{|{\cal M}_2|} \bol{1}[v_i 
= v] \right] \nonumber \\
 &=& P_V(v).
 \label{eq:average-c2-u}
\end{eqnarray}
Furthermore, 
for fixed ${\cal C}_2$, by taking the average over the randomly generated code ${\cal C}_1$, 
we have
\begin{eqnarray}
\lefteqn{ \mathbb{E}_{{\cal C}_1} \left[ P_{\tilde{X}|\tilde{V}}(x|v) \right] } \nonumber \\
&=& \mathbb{E}_{{\cal C}_1}\left[ \frac{\sum_{i \in {\cal M}_2} \sum_{j \in {\cal M}_1} 
\frac{1}{|{\cal M}_1| |{\cal M}_2|}
  \bol{1}[v_i = v,x_{ij} = x]}{P_{\tilde{V}}(v)} \right] \nonumber \\
&=& \frac{\sum_{i \in {\cal M}_2} \sum_{j \in {\cal M}_1} \frac{1}{|{\cal M}_1| |{\cal M}_2|} 
\bol{1}[v_i = v] P_{X|V}(x|v)}{P_{\tilde{V}(v)}} \nonumber \\
&=& P_{X|V}(x|v) 
\label{eq:average-c1-xu}
\end{eqnarray}

Let $P_{Z^\prime}$ be the output distribution when the input distribution is
$P_{\tilde{V}}(v)P_{X|V}(x|v)$. Then, from Eq.~(\ref{eq:average-c1-xu}), we have
\begin{eqnarray*}
\mathbb{E}_{{\cal C}_1} \left[ P_{\tilde{Z}}(z) \right] = P_{Z^\prime}(z)
\end{eqnarray*}
for every $z \in {\cal Z}$. Thus, we have
\begin{eqnarray}
\lefteqn{
\mathbb{E}_{{\cal C}_1 {\cal C}_2} \left[ D(P_{\tilde{Z}} \| P_Z) \right] } \nonumber \\
&=& \mathbb{E}_{{\cal C}_1 {\cal C}_2} \left[
   \sum_z P_{\tilde{Z}}(z) \log \frac{P_{\tilde{Z}}(z)}{P_Z(z)} \right] \nonumber \\
&=& \mathbb{E}_{{\cal C}_2} \left[
	\mathbb{E}_{{\cal C}_1}\left[ \sum_z P_{\tilde{Z}}(z) \log \frac{P_{\tilde{Z}}(z)}{P_{Z^
\prime}(z)} \right] \right. \nonumber \\
&&	~~~~+ \left. \mathbb{E}_{{\cal C}_1}\left[ \sum_z P_{\tilde{Z}}(z) \log \frac{P_{Z^
\prime}(z)}{P_Z(z)} \right]
\right] \nonumber \\
&=& \mathbb{E}_{{\cal C}_1 {\cal C}_2} \left[ \sum_z P_{\tilde{Z}}(z) \log \frac{P_{\tilde{Z}}
(z)}{P_{Z^\prime}(z)} \right] \nonumber \\
&&	~~~~+ \mathbb{E}_{{\cal C}_2}\left[ \sum_z P_{Z^\prime}(z) \log \frac{P_{Z^\prime}
(z)}{P_Z(z)} \right] \nonumber \\
&=& \mathbb{E}_{{\cal C}_1 {\cal C}_2}\left[ D(P_{\tilde{Z}} \| P_{Z^\prime}) \right]
	+ \mathbb{E}_{{\cal C}_2}\left[ D(P_{Z^\prime} \| P_Z) \right].
	\label{eq:proof-divergence-bound-0}
\end{eqnarray}

We bound each term of Eq.~(\ref{eq:proof-divergence-bound-0})
by using Proposition \ref{proposition:single-resolvability}.
By the monotonicity of the divergence, we have
\begin{eqnarray*}
D(P_{\tilde{Z}} \| P_{Z^\prime})
&\le& D(P_{I \tilde{Z}} \| P_{I Z^\prime} ) \\
&=& \sum_{i \in {\cal M}_2} \frac{1}{|{\cal M}_2|} D(P_{\tilde{Z}|I}(\cdot|i) \| P_{Z|V}
(\cdot|v_i)),
\end{eqnarray*}
where 
\begin{eqnarray*}
P_{I\tilde{Z}}(i,z) = \frac{1}{|{\cal M}_2|} \sum_{j \in {\cal M}_1} \frac{1}{|{\cal M}_1|} P_{Z|X}
(z|x_{ij})
\end{eqnarray*}
and 
\begin{eqnarray*}
P_{I Z^\prime}(i,z) = \frac{1}{|{\cal M}_2|} P_{Z|V}(z|v_i).
\end{eqnarray*}
Thus, by taking average over ${\cal C}_1$ and by using 
Proposition \ref{proposition:single-resolvability}
for input distribution $P_{X|V}(\cdot|v_i)$ instead of $P_X$, we have
\begin{eqnarray*}
\lefteqn{ \mathbb{E}_{{\cal C}_1}\left[ D(P_{\tilde{Z}} \| P_{Z^\prime}) \right] } \\
	&\le& \sum_{i \in {\cal M}_2} \frac{1}{|{\cal M}_2|} 
		\mathbb{E}_{{\cal C}_1}\left[ D(P_{\tilde{Z}|I}(\cdot|i) \| P_{Z|V}(\cdot|v_i)) \right] 
\\
	&\le& \sum_{i \in {\cal M}_2} \frac{1}{|{\cal M}_2|} \frac{1}{\theta |{\cal M}_1|^\theta}
		e^{\psi(\theta|P_{Z|X},P_{X|V}(\cdot|v_i))}. 
\end{eqnarray*}
By taking average over ${\cal C}_2$ and by noting Eq.~(\ref{eq:average-c2-u}), we have
\begin{eqnarray*}
\lefteqn{ \mathbb{E}_{{\cal C}_1 {\cal C}_2}\left[ D(P_{\tilde{Z}} \| P_{Z^\prime}) \right] } \\
	&\le& \sum_v P_V(v) \frac{1}{\theta |{\cal M}_1|^\theta} e^{\psi(\theta|P_{Z|X},P_{X|V}
(\cdot|v))} \\
	&=& \frac{1}{\theta |{\cal M}_1|^\theta} e^{\psi(\theta|P_{Z|X},P_{X|V},P_V)}.
\end{eqnarray*} 

On the other hand, by using Proposition \ref{proposition:single-resolvability}
for input distribution $P_V$ and channel $P_{Z|V}$, we have
\begin{eqnarray*}
\mathbb{E}_{{\cal C}_2}\left[ D(P_{Z^\prime} \| P_Z) \right]
	\le \frac{1}{\theta^\prime |{\cal M}_2|^{\theta^\prime}} e^{\psi(\theta^
\prime|P_{Z|V},P_V)}.
\end{eqnarray*}
\qed

%%%%% Proof of Corollary %%%%%%%%%%
\subsection{Proof of Corollary \ref{corollary:resolvability}}
\label{proof-of-corollary:resolvability}

We can choose $\gamma > 0$ such that
$R_1 \ge I(X;Z|V) + 2 \gamma$ and $R_2 \ge I(V;Z) + 2 \gamma$.
Let $|{\cal M}_{1,n}| = \lfloor e^{nR_1} \rfloor$, $|{\cal M}_{2,n}| = \lfloor e^{nR_2} 
\rfloor$.
Since $\psi^\prime(0| P_{Z|X},P_{X|V},P_V) = I(X;Z|V)$, there
exists $\theta_0 > 0$ such that
\begin{eqnarray*}
\frac{\psi(\theta_0 | P_{Z|X},P_{X|V},P_V)}{\theta_0} \le
 I(X;Z|V) + \gamma \le R_1 - \gamma,
\end{eqnarray*} 
which implies
\begin{eqnarray*}
- \frac{\theta_0}{n} \log |{\cal M}_{1,n}| + \psi(\theta_0 | P_{Z|X}, P_{X|V},P_V) \le - 
\gamma < 0.
\end{eqnarray*}
Thus, the first term of Eq.~(\ref{eq:resolvability-bound-divergence}) converges to $0$
asymptotically.
Similarly, we can show that the second term of Eq.~(\ref{eq:resolvability-bound-divergence})
converges to $0$ asymptotically. 
Thus, we have the assertion of the corollary. \qed

%%%% Rate Spliting %%%%%
\subsection{Proof of Lemma \ref{lemma:rate-spliting}}
\label{appendix:rate-spliting}

Although the lemma can be systematically proved by
using the Fourier-Motzkin elimination, we explicitly
find $(r_d,r_0,r_s)$ satisfying
\begin{eqnarray}
\label{eq:condition-inner}
(R_d - r_d,R_0 + r_0,R_1 - r_0 - r_s + r_d, R_s + r_s) \in {\cal R}^{(in)}
\end{eqnarray}
for given $(R_d,R_0,R_1,R_s) \in {\cal R}^*$ as follows.

If $R_1 + R_s \le I(V; Y|U)$ and $R_1 < I(V;Z|U)$, we set
\begin{eqnarray*}
r_d := I(V;Z|U) - R_1
\end{eqnarray*}
and $(r_0,r_s) := (0,0)$.
Then, Eqs.~(\ref{eq:inner-1}), and (\ref{eq:inner-4})
are obviously satisfied. Eq.~(\ref{eq:inner-2}) can be confirmed as
\begin{eqnarray*}
R_1 + r_d + R_s 
	&\le& I(V;Z|U) + I(V;Y|U) - I(V;Z|U) \\
	&=& I(V;Y|U),
\end{eqnarray*}
Eq~(\ref{eq:inner-3}) can be confirmed as 
\begin{eqnarray*}
R_0 + R_1 + r_d + R_s 
	&\le& I(U;Y) + I(V;Y|U) \\
	&=& I(V;Y),
\end{eqnarray*}
and Eq.~(\ref{eq:inner-6}) can be confirmed as
\begin{eqnarray*}
R_d - r_d 
	&=& R_d + R_1 - I(V;Z|U) \\
	&\ge& I(X;Z|U) - I(V;Z|U) \\
	&=& I(X;Z|V).
\end{eqnarray*}
Thus, Eq.~(\ref{eq:condition-inner}) holds. 

If $R_1 + R_s \le I(V;Y|U)$ and $R_1 \ge I(V; Z|U)$,
we set $(r_d,r_1,r_s) := (0,0,0)$.
Then, Eq.~(\ref{eq:condition-inner}) obviously holds.

If $R_1 + R_s > I(V; Y|U)$, 
we set $r_d := 0$ and
\begin{eqnarray}
r_s &:=& I(V; Y|U) - I(V;Z|U) - R_s, \nonumber \\
r_0 &:=& R_1 + R_s - I(V;Y|U) \nonumber \\
	&=& (R_1 - r_s) + (R_s + r_s) - I(V;Y|U) \nonumber \\
	&=& R_1 - r_s - I(V;Z|U).
	\label{eq:rate-spliting}
\end{eqnarray}
Then, Eqs.~(\ref{eq:inner-2}), (\ref{eq:inner-3}), 
and (\ref{eq:inner-6}) are obviously satisfied.
Eq.~(\ref{eq:inner-1}) can be confirmed as
\begin{eqnarray*}
R_0 + r_0
	&=& R_0 + R_1 + R_s - I(V;Y|U) \\
	&\le& I(U;Z),
\end{eqnarray*}
and Eq.~(\ref{eq:inner-4}) can be confirmed from Eq.~(\ref{eq:rate-spliting}).
Thus, Eq.~(\ref{eq:condition-inner}) is satisfied.

%%%%%% Error Analysis %%%%
\subsection{Proof of Lemma \ref{lemma:error-analysis}}
\label{proof-lemma:error-analysis}

% \label{eq:error-analysis-1}

\paragraph{Proof of Eq.~(\ref{eq:error-analysis-1})}
We first note the following observations. By taking the average over
randomly generated codes, we have
\begin{eqnarray}
\lefteqn{\mathbb{E}_{{\cal C}_0 {\cal C}_1 {\cal C}_2}[P_{err}(f,g)] } \nonumber \\
	&=& \mathbb{E}_{{\cal C}_0 {\cal C}_1 {\cal C}_2} \left[ 
		\sum_{k,\ell,s,a} \frac{1}{|{\cal K}||{\cal L}||{\cal S}||{\cal A}|} P_{Y|X}({\cal D}_{k
\ell s}^c|x_{k\ell s a}) \right] \nonumber \\
	&=& \mathbb{E}_{{\cal C}_0  {\cal C}_2} \left[ 
		\sum_{k,\ell,s,a} \frac{1}{|{\cal K}||{\cal L}||{\cal S}||{\cal A}|} \mathbb{E}_{{\cal 
C}_1}\left[P_{Y|X}({\cal D}_{k\ell s}^c|x_{k\ell s a})\right] \right] \nonumber \\
	&=& \mathbb{E}_{{\cal C}_0  {\cal C}_2} \left[ 
		\sum_{k,\ell,s} \frac{1}{|{\cal K}||{\cal L}||{\cal S}|} P_{Y|V}({\cal D}_{k\ell s}^c|v_{k
\ell s}) \right]. 
		\label{eq:proof-error-analysis-1}
\end{eqnarray}
Let ${\cal T}_{uv} = \{ y : (u,v,y) \in {\cal T} \}$. Then, we have
\begin{eqnarray*}
\lefteqn{ \mathbb{E}_{{\cal C}_0  {\cal C}_2} \left[ 
		\sum_{k,\ell,s} \frac{1}{|{\cal K}||{\cal L}||{\cal S}|} P_{Y|V}({\cal D}_{k\ell s}^c|v_{k
\ell s}) \right] } \\
	&\le& \mathbb{E}_{{\cal C}_0 {\cal C}_1} \left[ \sum_{k,\ell,s} \frac{1}{|{\cal K}||{\cal 
L}||{\cal S}|} \{
		P_{Y|V}({\cal T}_{u_k v_{k\ell s}}^c|v_{k\ell s}) \right.  \\
	&& \left. + \sum_{(\hat{k},\hat{\ell},\hat{s}) \neq (k,\ell,s)} P_{Y|V}({\cal 
T}_{u_{\hat{k}} v_{\hat{k}\hat{\ell}\hat{s}}}|v_{k\ell s}) \} \right] \\
	&=& \mathbb{E}_{{\cal C}_0 {\cal C}_2} \left[ \sum_{k,\ell,s} \frac{1}{|{\cal K}||{\cal 
L}||{\cal S}|} \left\{
		P_{Y|V}({\cal T}_{u_k v_{k\ell s}}^c|v_{k\ell s}) \right. \right. \\
	&&  + \sum_{(\hat{\ell},\hat{s}) \neq (\ell,s)} P_{Y|V}({\cal T}_{u_k v_{k \hat{\ell} 
\hat{s}}}|v_{k \ell s}) \\
	&&  \left. + \sum_{k \neq \hat{k}} \sum_{\hat{\ell},\hat{s}}  P_{Y|V}({\cal 
T}_{u_{\hat{k}} v_{\hat{k}\hat{\ell}\hat{s}}}|v_{k\ell s}) \} \right] \\
	&\le& \sum_{k,\ell,s} \frac{1}{|{\cal K}||{\cal L}||{\cal S}|} \{
		P_{UVY}({\cal T}^c) \\
	&&	+ |{\cal L}||{\cal S}| \sum_{u,v} P_{UV}(u,v) P_{Y|U}({\cal T}_{uv}|u) \\
	&&	+ |{\cal K}||{\cal L}||{\cal S}| \sum_{u,v}P_{UV}(u,v) P_Y({\cal T}_{uv}) \} \\
	&\le& P_{UVY}({\cal T}^c) + |{\cal L}||{\cal S}|e^{-\alpha_1} + |{\cal K}||{\cal L}||{\cal 
S}| e^{-\alpha_2},
\end{eqnarray*}
where we used
\begin{eqnarray*}
P_{Y|U}(y|u) &\le& P_{Y|V}(y|v) e^{- \alpha_1}, \\
P_Y(y) &\le& P_{Y|V}(y|v) e^{- \alpha_2}
\end{eqnarray*}
for $y \in {\cal T}_{uv}$ in the last inequality.

%%%%
\paragraph{Proof of Eq.~(\ref{eq:error-analysis-2})}
In a similar manner as Eq.~(\ref{eq:proof-error-analysis-1}), we have
\begin{eqnarray*}
\mathbb{E}_{{\cal C}_0 {\cal C}_1 {\cal C}_2} \left[ P_{err}(f,\phi) \right]
	= \mathbb{E}_{{\cal C}_0 } \left[ \sum_{k} \frac{1}{|{\cal K}|} P_{Z|U}({\cal D}_k^c|u_k) 
\right],
\end{eqnarray*}
which is just a random coding error probability of channel $P_{Z|U}$.
Thus, by the standard arguments of the information spectrum approach \cite{han:book},
we have Eq.~(\ref{eq:error-analysis-2}).

%%%%
\paragraph{Proof of Eq.~(\ref{eq:error-analysis-3})}

By the monotonicity of the divergence, we have
\begin{eqnarray*}
\lefteqn{ D(P_{S \tilde{Z}} \| P_S \times P_{\tilde{Z}}) } \\
	&\le& D(P_{K S \tilde{Z}} \| P_S \times P_{K \tilde{Z}} ) \\
	&=& \sum_k \frac{1}{|{\cal K}|} D(P_{S\tilde{Z}|K}(\cdot|k) \| P_S \times P_{\tilde{Z}|K}
(\cdot|k)) \\
	&=& \sum_{k,s} \frac{1}{|{\cal K}||{\cal S}|} D(P_{\tilde{Z}|KS}(\cdot|k,s) \| 
P_{\tilde{Z}|K}(\cdot|k)).
\end{eqnarray*}
Note that the relation
\begin{eqnarray*}
\lefteqn{ \sum_s \frac{1}{|{\cal S}|} D(P_{\tilde{Z}|KS}(\cdot|k,s) \| P_{\tilde{Z}|K}(\cdot|k)) } 
\\
	&& + D(P_{\tilde{Z}|K}(\cdot|k) \| P_{Z|U}(\cdot|u_k))  \\
	&=& \sum_s \frac{1}{|{\cal S}|} D(P_{\tilde{Z}|KS}(\cdot|k,s) \| P_{Z|U}(\cdot|u_k))
\end{eqnarray*}
holds for each $k \in {\cal K}$. Thus, by using Lemma \ref{theorem:resolvability} for
$P_{V|U}(\cdot|u_k)$ instead of $P_V$, we have
\begin{eqnarray*}
\lefteqn{ \mathbb{E}_{{\cal C}_1 {\cal C}_2}\left[ D(f) \right] } \\
	&\le& \sum_k \frac{1}{|{\cal K}|} \left[ 
		\frac{1}{\theta |{\cal A}|^\theta} e^{\psi(\theta|P_{Z|X},P_{X|V}, P_{V|U}
(\cdot|u_k))} \right. \\
	&& \left. + \frac{1}{\theta^\prime |{\cal L}|^{\theta^\prime}} e^{\psi(\theta^
\prime|P_{Z|V}, P_{V|U}(\cdot|u_k))} \right].
\end{eqnarray*}
By taking average over ${\cal C}_0$ and by noting 
\begin{eqnarray*}
\mathbb{E}_{{\cal C}_0}\left[ \sum_k \frac{1}{|{\cal K}|} \bol{1}[u_k = u] \right] = P_U(u),
\end{eqnarray*}
we have
\begin{eqnarray*}
\lefteqn{ \mathbb{E}_{{\cal C}_0 {\cal C}_1 {\cal C}_2}\left[ D(f) \right] } \\
	&\le& \sum_u P_U(u) \left[ 
		\frac{1}{\theta |{\cal A}|^\theta} e^{\psi(\theta|P_{Z|X},P_{X|V}, P_{V|U}
(\cdot|u))} \right. \\
	&& \left. + \frac{1}{\theta^\prime |{\cal L}|^{\theta^\prime}} e^{\psi(\theta^
\prime|P_{Z|V}, P_{V|U}(\cdot|u))} \right] \\ 
	&=& \frac{1}{\theta |{\cal A}|^\theta} e^{\psi(\theta|P_{Z|X},P_{X|V},P_{V})} \nonumber 
\\
	&& +\frac{1}{\theta^\prime |{\cal L}|^{\theta^\prime}} e^{\psi(\theta^
\prime|P_{Z|V},P_{V|U},P_U)}.
\end{eqnarray*}
\qed

%%%%%
\subsection{Proof of Corollary \ref{corollary-more-capable}}
\label{proof-of-corollary-main}

By noting that $U$, $V$, $X$, and $(Y,Z)$ form Markov chain, 
which implies $I(V; Y|X,U) = 0$ and $I(V;Z|X,U) = 0$, we have
\begin{eqnarray*}
\lefteqn{
I(V;Y|U) - I(V;Z|U)
} \\
&=& I(V,X; Y|U) - I(V,X; Z|U) \\
	&& - [ I(X; Y|U,V) - I(X;Z|U,V)] \\
&=& I(X; Y|U) - I(X; Z|U) \\
	&& - [ I(X; Y|U,V) - I(X; Z|U,V)].
\end{eqnarray*}
Since $P_{Y|X}$ is more capable than $P_{Z|X}$, we have
\begin{eqnarray*}
I(X; Y|U = u, V=v) - I(X; Z|U=u, V=v) \ge 0
\end{eqnarray*}
for every $(u,v)$, which implies
\begin{eqnarray*}
I(V;Y|U) - I(V;Z|U) \le I(X;Y|U) - I(X;Z|U).
\end{eqnarray*}
Thus, the auxiliary random variable $V$ is not needed.
%The statement on the range size of $U$ can be 
%proved by the standard support lemma argument \cite{csiszar-korner:11}.
\qed

%\bibliographystyle{../09-04-17-bibtex/IEEEtran}
%\bibliography{../09-04-17-bibtex/reference.bib}


%%%%% Author Bib %%%%%%%%%%%%%%%%%

%\begin{IEEEbiography}{Shun Watanabe}
%received the B.E., 
%M.E., and Ph.D.\ degrees from Tokyo Institute of Technology
%in 2005, 2007, and 2009 respectively. He is currently
%an Assistant Professor in the Department of Information 
%Science and Intelligent Systems of  University of Tokushima.
%His current research interests are in the areas of
%information theory, quantum information theory,
%and quantum cryptography.
%\end{IEEEbiography}

\end{document}